\documentclass[a4paper,10pt]{article}
\usepackage{test,paralist,graphicx,subcaption,booktabs,caption,multirow}
\usepackage[authoryear,round]{natbib}
\usepackage{framed}

\usepackage{hyperref}
\usepackage{doi}

\numberwithin{equation}{section}

\title{Data-based Automatic Discretization of Nonparametric Distributions}
\author{Alexis Akira Toda\thanks{Department of Economics, University of California San Diego. Email: \href{mailto:atoda@ucsd.edu}{atoda@ucsd.edu}.}}

\begin{document}
\maketitle

\begin{abstract}
Although using non-Gaussian distributions in economic models has become increasingly popular, currently there is no systematic way for calibrating a discrete distribution from the data without imposing parametric assumptions. This paper proposes a simple nonparametric calibration method based on the \cite{GolubWelsch1969} algorithm for Gaussian quadrature. 
Application to an optimal portfolio problem suggests that assuming Gaussian instead of nonparametric shocks leads to up to 17\% overweighting in the stock portfolio because the investor underestimates the probability of crashes.

\medskip


\medskip

\end{abstract}

\section{Introduction}

This paper studies the following problem, which applied theorists often encounter. A researcher would like to calibrate the parameters of a stochastic model. One of the model inputs is a probability distribution of shocks, which is to be approximated by a discrete distribution. Due to computational considerations, the researcher would like this distribution to have as few support points (nodes) as possible, say five. Given the data of shocks, how should the researcher calibrate the nodes and probabilities of this five-point distribution?

While there are many established methods for discretizing processes with Gaussian shocks such as \cite{tauchen1986-EL}, \cite{tauchen-hussey1991}, and \cite{rouwenhorst1995},\footnote{See \cite{FarmerToda2017QE} and the references therein for a detailed literature review.} discretizing non-Gaussian distributions remains relatively unexplored. However, it has become increasingly common in economics to study models with non-Gaussian shocks. For example, the rare disasters model \citep{rietz1988,Barro2006QJE,Gabaix2012RareDisaster} uses rare but large downward jumps to explain asset pricing puzzles. One issue with discretizing non-Gaussian distributions is how to calibrate them. If we have a parametric density, it is possible to discretize it using the Gaussian quadrature as in \cite{miller-rice1983} or the maximum entropy method as in \cite{TanakaToda2013EL,TanakaToda2015SINUM} provided that we can compute some moments. However, it is not obvious how to obtain an $N$-point distribution that approximates the data well without imposing parametric assumptions. Because the degree of freedom of an $N$-point distribution is large ($2N-1$), providing an automatic discretization method is valuable because it removes the arbitrariness of calibration.

Given the data, this paper proposes a simple method for automatically calibrating a discrete distribution with a specified number of grid points. The method is based on the observation that to compute the nodes and weight of the $N$-point Gaussian quadrature with some weighting function using the \cite{GolubWelsch1969} algorithm, one only needs to know the moments of the weighting function up to order $2N$. Therefore a natural way to discretize a nonparametric distribution is simply to feed the $2N$ sample moments into the Golub-Welsch algorithm. Since this method does not involve optimization (it is a matter of solving for the eigenvalues/vectors of a sparse matrix), the implementation is easy and fast.

As an application, I discretize the U.S.\ historical stock returns data and solve an optimal portfolio problem with constant relative risk aversion utility. I consider two cases in which the investor uses the nonparametric and Gaussian densities. I show that when the investor incorrectly believes that the stock returns distribution is lognormal, the stock portfolio is overweighted by up to 17\% because he underestimates the probability of crashes. These examples show that the choice of the calibration method may matter quantitatively.

\subsection{Related literature}
The closest paper to mine is \cite{miller-rice1983}, who use the Gaussian quadrature to discretize distributions. While they consider only the discretization of parametric distributions, my focus is on the discretization of nonparametric distributions estimated from data. \cite{TanakaToda2013EL} consider the discretization of distributions on preassigned nodes by matching the moments using the maximum entropy principle, and \cite{TanakaToda2015SINUM} prove convergence and obtain an error estimate. \cite{FarmerToda2017QE} consider the discretization of general non-Gaussian Markov processes by applying the Tanaka-Toda method to conditional distributions. In one of the applications, they discretize a nonparametric density on a preassigned grid by approximating it with a Gaussian mixture. Since computing the nodes and weights of Gaussian quadrature does not require optimizing over parameters (unlike the maximum likelihood estimation of Gaussian mixture parameters or solving the maximum entropy problem), my method is easier and faster to implement, and the grid is chosen endogenously. On the other hand, the \cite{FarmerToda2017QE} method can discretize general Markov processes, whereas the proposed method in this paper is designed to discretize a single distribution.

\section{Discretizing a nonparametric density}\label{sec:method}

Suppose for the moment that the nonparametric density $f(x)$ is known. Since stochastic models often involve expectations, we would like to find nodes $\set{x_n}_{n=1}^N$ and weights $\set{w_n}_{n=1}^N$ such that
\begin{equation}
\E[g(X)]=\int_{-\infty}^\infty g(x)f(x)\diff x\approx \sum_{n=1}^Nw_ng(x_n),\label{eq:QuadForm}
\end{equation}
where $g$ is a general integrand and $X$ is a random variable with density $f(x)$. The right-hand side of \eqref{eq:QuadForm} defines an $N$-point quadrature formula.

When \eqref{eq:QuadForm} is exact (\ie, $\approx$ becomes $=$) for all polynomials of degree up to $D$, we say that the quadrature formula has degree of exactness $D$. Since the degree of freedom in an $N$-point quadrature formula is $2N$ (because there are $N$ nodes and $N$ weights), we cannot expect to integrate more than $2N$ monomials $f(x)=1,x,\dots,x^{2N-1}$ exactly. When the quadrature formula \eqref{eq:QuadForm} is exact for these monomials, or equivalently when it has degree of exactness $2N-1$, we call the formula \emph{Gaussian}. The following Golub-Welsch algorithm provides an efficient way to compute the nodes and weights of the Gaussian quadrature. (Appendix \ref{sec:GQ} provides more theoretical background.)

\begin{framed}
\begin{algorithm}[\citealp{GolubWelsch1969}]\label{alg:GW}
\quad
\begin{enumerate}
\item Select a number of quadrature nodes $N\in \N$.
\item For $k=0,1,\dots,2N$, compute the $k$-th moment of the density $m_k=\int x^kf(x)\diff x$.
\item Define the matrix of moments $M=(M_{ij})_{1\le i,j\le N+1}$ by $M_{ij}=m_{i+j-2}$.
\item Compute the Cholesky factorization $M=R'R$. Let $R=(r_{ij})_{1\le i,j\le N+1}$.
\item Define $\alpha_1=r_{12}/r_{11}$, $\alpha_n=\frac{r_{n,n+1}}{r_{nn}}-\frac{r_{n-1,n}}{r_{n-1,n-1}}$ ($n=2,\dots,N$), and $\beta_n=\frac{r_{n+1,n+1}}{r_{nn}}$ ($n=1,\dots,N-1$). Define the $N\times N$ symmetric tridiagonal matrix
\begin{equation}
T_N=\begin{bmatrix}
\alpha_1 & \beta_1 & 0 & \cdots & 0\\
\beta_1 & \alpha_2 & \beta_2 & \ddots & \vdots \\
0 & \beta_2 & \alpha_3 & \ddots & 0 \\
\vdots & \ddots & \ddots & \ddots & \beta_{N-1} \\
0 & \cdots & 0 & \beta_{N-1} & \alpha_N
\end{bmatrix}.\label{eq:TN}
\end{equation}
\item Compute the eigenvalues $\set{x_n}_{n=1}^N$ of $T_N$ and the corresponding eigenvectors $\set{v_n}_{n=1}^N$. $\set{x_n}_{n=1}^N$ are the nodes of the Gaussian quadrature and the weights $\set{w_n}_{n=1}^N$ are given by $w_n=m_0v_{n1}^2/\norm{v_n}^2>0$, where $v_n=(v_{n1},\dots,v_{nn})'$.
\end{enumerate}
\end{algorithm}
\end{framed}

Once we compute the nodes $\set{x_n}_{n=1}^N$ and weights $\set{w_n}_{n=1}^N$, we can use them as the discrete approximation of the density $f$.

Note that the only inputs to the Golub-Welsch algorithm \ref{alg:GW} are the number of nodes $N$ and the moments $m_k=\int x^kf(x)\diff x$ of the density $f$, where $k=0,\dots,2N$. Therefore a natural idea for discretizing a nonparametric density given the data $\set{x_i}_{i=1}^I$ is simply to feed the sample moments into the Golub-Welsch algorithm \ref{alg:GW}. Summarizing the above observations, we obtain the following algorithm for the data-based automatic discretization of nonparametric distributions.

\begin{framed}
\begin{algorithm}[Automatic discretization of nonparametric distributions]\label{alg:disc}
\quad
\begin{enumerate}
\item Given the data $\set{x_i}_{i=1}^I$ and the desired number of discrete points $N$, for $k=0,\dots,2N$ compute the $k$-th sample moment
\begin{equation}
\widehat{m}_k=\frac{1}{I}\sum_{i=1}^Ix_i^k.\label{eq:samplemoment}
\end{equation}
\item Feed these moments $\set{\widehat{m}_k}_{k=0}^{2N}$ into the Golub-Welsch algorithm \ref{alg:GW} to compute the nodes $\set{\bar{x}_n}_{n=1}^N$ and weights $\set{w_n}_{n=1}^N$. The desired discretization assigns probability $w_n$ on the point $\bar{x}_n$.
\end{enumerate}
\end{algorithm}
\end{framed}

Because the $N$-point Gaussian quadrature has degree of exactness $2N-1$, by construction the $N$-point discretization matches the sample moments of data up to order $2N-1$ (and up to numerical error). Since by the Gauss-Markov theorem the sample moment \eqref{eq:samplemoment} is the best linear unbiased estimator (BLUE) of the population moment, that is, $\widehat{m}_k$ has the minimum mean-squared error among all estimates of the form $\sum_{i=1}^Ia_ix_i^k$, where $\set{a_i}_{i=1}^I$ are some weights, Algorithm \ref{alg:disc} is in a sense optimal.

Appendix \ref{sec:accuracy} shows that the accuracy of the proposed method exceeds that of using a parametric distribution when the latter is misspecified.

\section{Application: optimal portfolio problem}\label{subsec:appl.2}

In this section I illustrate the usefulness of the proposed method using minimal economic examples.
Consider a CRRA investor with relative risk aversion $\gamma>0$. Letting $R>0$ be the gross stock return, $R_f>0$ be the gross risk-free rate, and $\theta$ be the fraction of wealth (portfolio share) invested in the stock, the investor's optimal portfolio problem is
\begin{equation}
\max_\theta \frac{1}{1-\gamma}\E[(R\theta+R_f(1-\theta))^{1-\gamma}].\label{eq:portprob}
\end{equation}

I obtain the annual data on U.S.\ nominal stock returns, risk-free rate, and inflation for the period 1927--2016 from the spreadsheet of Amit Goyal.\footnote{The spreadsheet is at \url{http://www.hec.unil.ch/agoyal/docs/PredictorData2016.xlsx}. Using monthly or quarterly data give qualitatively similar results, though slightly less extreme quantitatively.} For the stock returns I use the CRSP volume-weighted index including dividends. I convert these returns into real log returns and calibrate the log risk free rate as the sample average. The result is $R_f=1.0045$. I then apply Algorithm \ref{alg:disc} to the log excess returns $\log R-\log R_f$ to obtain a discrete distribution with nodes $\set{\bar{x}_n}_{n=1}^N$ and weights $\set{w_n}_{n=1}^N$, where I choose the number of points $N=5$ (increasing the number of points further does not change the results). The gross stock return in state $n$ is defined by $R_n=R_f\e^{\bar{x}_n}$, which occurs with probability $w_n$. Finally, I numerically solve the optimal portfolio problem \eqref{eq:portprob}. Figure \ref{fig:portprob} shows the results when we change the relative risk aversion in the range $\gamma\in [1,7]$.

\begin{figure}[!htb]
\centering
\begin{subfigure}{0.48\linewidth}
\includegraphics[width=\linewidth]{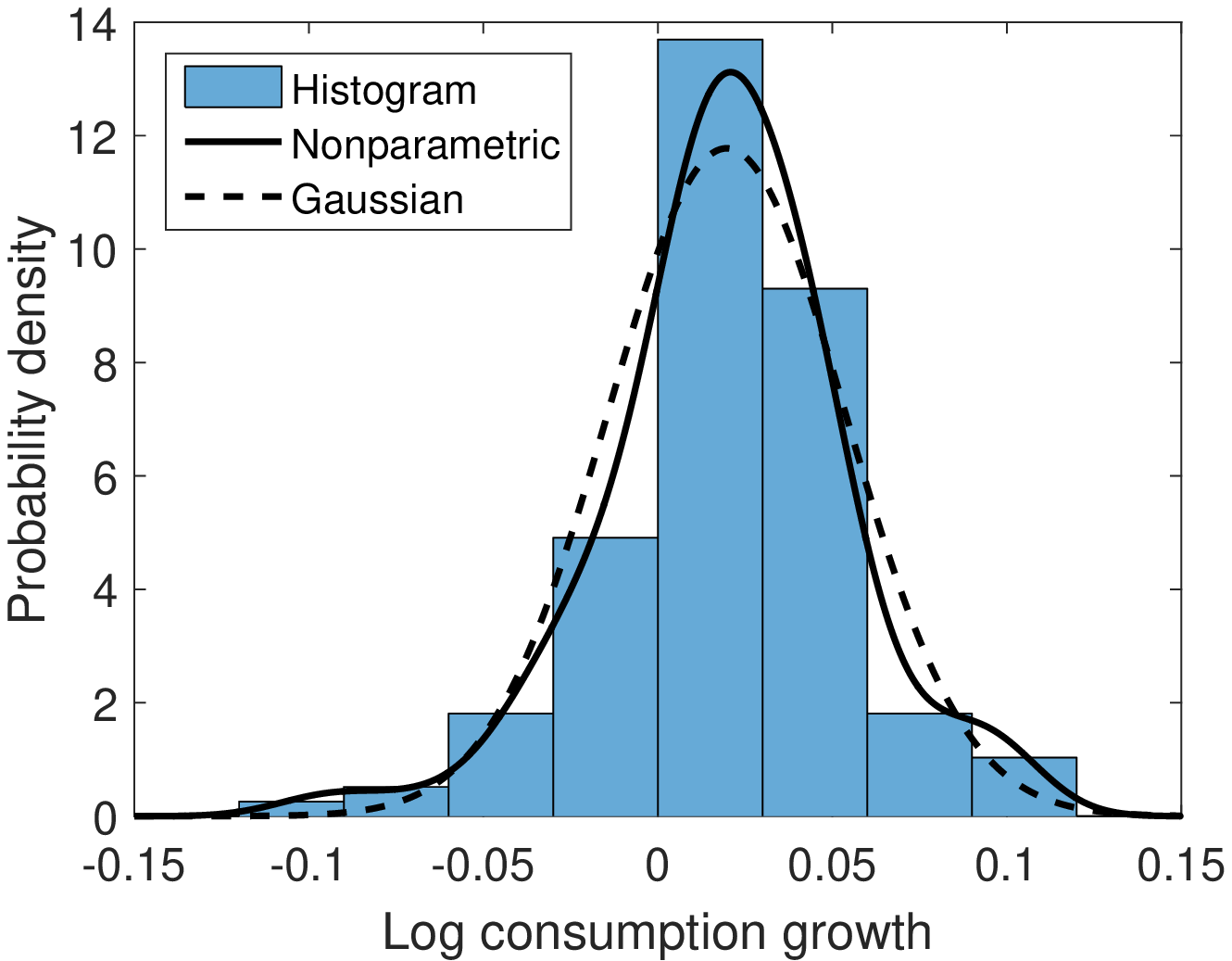}
\caption{Histogram and densities.}\label{fig:hist}
\end{subfigure}
\begin{subfigure}{0.48\linewidth}
\includegraphics[width=\linewidth]{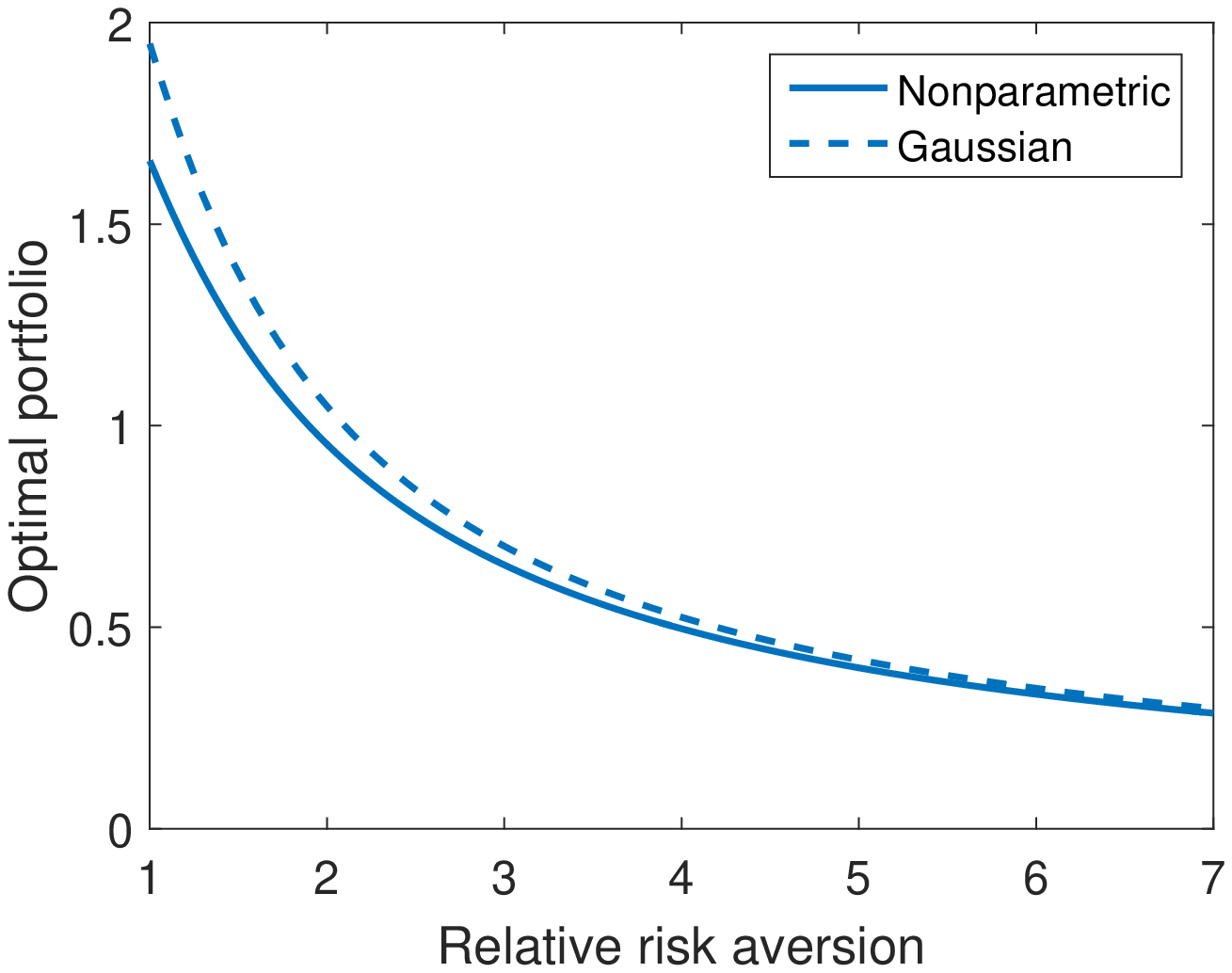}
\caption{Optimal portfolio.}\label{fig:portfolio}
\end{subfigure}
\begin{subfigure}{0.48\linewidth}
\includegraphics[width=\linewidth]{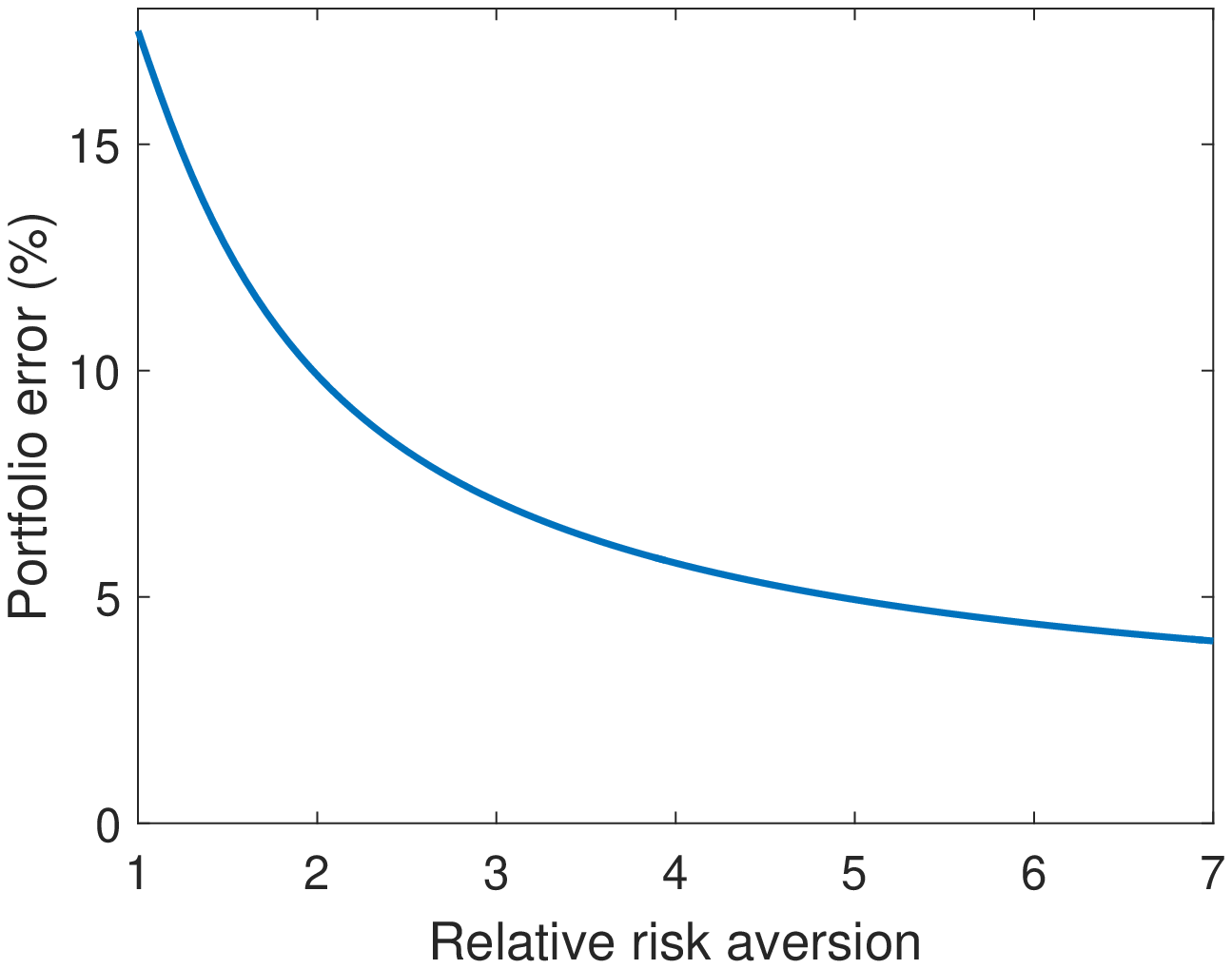}
\caption{Portfolio error.}\label{fig:porterror}
\end{subfigure}
\caption{Excess returns distribution and numerical solution.}\label{fig:portprob}
\end{figure}

Figure \ref{fig:hist} shows the histogram of the log excess returns distribution as well as the nonparametric kernel density estimator and the Gaussian distribution fitted by maximum likelihood. We can see that the histogram and the nonparametric density have a long left tail corresponding to stock market crashes, which the Gaussian distribution misses. Figure \ref{fig:portfolio} shows the optimal portfolio $\theta$ for the two models. We can see that when the investor incorrectly believes that the stock returns distribution is lognormal, he overweights the stock portfolio because he underestimates the probability of crashes. Figure \ref{fig:porterror} shows the percentage of this overweight (portfolio error) $\theta_\mathrm{G}/\theta_\mathrm{NP}-1$, where G and NP stand for Gaussian and nonparametric densities. The portfolio error is substantial, in the range of 4--17\%.

\section{Concluding remarks}

This paper has proposed a simple, automatic method for discretizing a nonparametric distribution, given the data. Using an asset pricing model and an optimal portfolio problem as a laboratory, I showed that the error from using a parametric distribution (such as the Gaussian distribution) can be substantial. 

A natural extension is to consider the discretization of Markov processes with nonparametric shocks. For example, one may be tempted to apply the kernel density estimation and Gaussian quadrature in the \cite{tauchen-hussey1991} method to discretize the AR(1) process
$$x_t=\rho x_{t-1}+\varepsilon_t,$$
where $\abs{\rho}<1$ and the innovations $\set{\varepsilon_t}_{t=0}^\infty$ are independent and identically distributed according to some probability density function $f$. However, it is well known that the Tauchen-Hussey method is not accurate when the persistence $\rho$ is moderately high \citep{floden2008}. Using the AR(1) asset pricing model in the Online Appendix of \cite{FarmerToda2017QE} to evaluate the solution accuracy, I found that the Gaussian quadrature-based methods for discretizing Markov processes is even less accurate when the shock distribution is nonparametric. Therefore for such processes, it is preferable to use the \cite{FarmerToda2017QE} maximum entropy method with an even-spaced grid (see their Section 4.3.3 for an example).

Finally, although I proposed my method as a tool for discretization, it can also be used as a quadrature method. For example, \cite{PohlSchmeddersWilms2018} solve the \cite{bansal-yaron2004} long run risks model using the projection method and Gauss-Hermite quadrature, but that is because the model is assumed to have Gaussian shocks. If instead a researcher wishes to use nonparametric shocks, my method can be directly used to construct a quadrature rule from data.



\appendix

\newpage

\noindent
{\Large \bf Online Appendix}

\section{Accuracy}\label{sec:accuracy}
As in any numerical method, evaluating the accuracy is very important. In this section I evaluate the accuracy of the proposed method using the optimal portfolio problem in Section \ref{subsec:appl.2} as a laboratory.

I design the numerical experiment as follows. First I fit a Gaussian mixture distribution with two components to the annual log excess returns data. The proportion, mean, and standard deviation of each mixture components are $p=(p_j)=(0.1392,0.8608)$, $\mu=(\mu_j)=(-0.2242,0.1064)$, and $\sigma=(\sigma_j)=(0.2164,0.1453)$, respectively. Next, I assume that the true excess returns distribution is this Gaussian mixture and solve the optimal portfolio problem for relative risk aversion $\gamma\in\set{2,4,6}$ using the Gaussian quadrature (Golub-Welsch algorithm \ref{alg:GW}) for Gaussian mixtures with 11 points. Finally, I generate random numbers from this Gaussian mixture with various sample sizes, discretize these distributions with various methods, and compute the optimal portfolio. I repeat this procedure with $M=\text{1,000}$ Monte Carlo replications and compute the relative bias and mean absolute error (MAE)
\begin{subequations}
\begin{align}
\mathrm{Bias}&=\frac{1}{M}\sum_{m=1}^M\left(\widehat{\theta}_m/\theta^*-1\right), \label{eq:bias}\\
\mathrm{MAE}&=\frac{1}{M}\sum_{m=1}^M\abs{\widehat{\theta}_m/\theta^*-1}, \label{eq:MAE}
\end{align}
\end{subequations}
where $\widehat{\theta}_m$ is the optimal portfolio from simulation $m$ and $\theta^*$ is the theoretical optimal portfolio. For the sample size I consider $T=100, \text{1,000}, \text{10,000}$, and for the number of quadrature nodes I consider $N=3,5,7,9$. For the discretization method I consider three cases. The first is the nonparametric Gaussian quadrature method (Algorithm \ref{alg:disc}), which I refer to as ``NP-GQ''. The second is the Gauss-Hermite quadrature, where the mean and standard deviation are estimated by maximum likelihood. This is the most natural method if the returns distribution is lognormal. The third is the maximum entropy method proposed by \cite{TanakaToda2013EL,TanakaToda2015SINUM} and \cite{FarmerToda2017QE} where the kernel density estimator (with Gaussian kernel) is fed into, which I refer to as ``NP-ME''. For this method one needs to assign the grid and the number of moments to match. Following Corollary 3.5 of \cite{FarmerToda2017QE}, I use an even-spaced grid centered at the sample mean that spans $\sqrt{2(N-1)}$ times the sample standard deviation at both sides, where $N$ is the number of grid points. I match 4 sample moments whenever possible for $N\ge 5$, and otherwise I match 2 sample moments (mean and variance). For more details on the exact algorithm, please refer to \cite{TanakaToda2013EL,TanakaToda2015SINUM} and Sections 2 and 3.2 of \cite{FarmerToda2017QE}.

Tables \ref{t:bias} and \ref{t:MAE} show the relative bias and mean absolute error of the optimal portfolio, respectively. As expected, the optimal portfolio computed using Gauss-Hermite is biased upwards because it uses the Gaussian distribution, which underestimates the probability of crashes. Among the two nonparametric discretization methods, NP-GQ uniformly outperforms NP-ME in terms of bias and mean absolute error, especially when the sample size is small ($T=100$). For $N=3$ grid points, in which case it is impossible to match 4 moments with NP-ME, the proposed NP-GQ method performs significantly better. Finally, increasing $N$ beyond 5 does not improve the bias or the mean absolute error for NP-GQ, which suggests that using a five-point distribution is enough (at least for solving this portfolio problem).

\begin{table}[!htb]
\centering
\caption{Relative bias of the optimal portfolio.}\label{t:bias}
\begin{tabular}{ccccccccccc}
\toprule
\multicolumn{2}{c}{Method} & \multicolumn{3}{c}{NP-GQ} & \multicolumn{3}{c}{Gauss-Hermite} & \multicolumn{3}{c}{NP-ME}\\
$T$ & $N$ & $\gamma=2$ & 4 & 6 & $\gamma=2$ & 4 & 6 & $\gamma=2$ & 4 & 6 \\
\cmidrule(lr){1-2}
\cmidrule(lr){3-5}
\cmidrule(lr){6-8}
\cmidrule(lr){9-11}
\multirow{4}{*}{100} & 3 & 0.054&0.053&0.053&0.168&0.123&0.109&0.140&0.113&0.105\\
& 5 & 0.051&0.053&0.053&0.159&0.123&0.109&0.096&0.082&0.077\\
& 7 & 0.051&0.053&0.053&0.158&0.123&0.109&0.082&0.074&0.070\\
& 9 & 0.051&0.053&0.053&0.157&0.123&0.109&0.076&0.071&0.068\\
\cmidrule(lr){1-2}
\cmidrule(lr){3-5}
\cmidrule(lr){6-8}
\cmidrule(lr){9-11}
\multirow{4}{*}{1,000} & 3 &0.005&0.005&0.005&0.105&0.060&0.047&0.089&0.059&0.050\\
& 5 &0.004&0.005&0.005&0.103&0.060&0.047&0.031&0.019&0.016\\
& 7 &0.004&0.005&0.005&0.103&0.060&0.047&0.022&0.014&0.011\\
& 9 &0.004&0.005&0.005&0.103&0.060&0.047&0.018&0.012&0.010\\
\cmidrule(lr){1-2}
\cmidrule(lr){3-5}
\cmidrule(lr){6-8}
\cmidrule(lr){9-11}
\multirow{4}{*}{10,000} & 3 & 0.001&0.001&0.001&0.098&0.054&0.041&0.084&0.054&0.045\\
& 5 &0.001&0.001&0.001&0.098&0.054&0.041&0.020&0.011&0.008\\
& 7 &0.001&0.001&0.001&0.098&0.054&0.041&0.012&0.006&0.004\\
& 9 &0.001&0.001&0.001&0.098&0.054&0.041&0.008&0.004&0.003\\
\bottomrule
\end{tabular}
\caption*{\footnotesize Note: the table reports the relative bias of the optimal portfolio defined by \eqref{eq:bias}. For discretization methods, ``NP-GQ'' uses Algorithm \ref{alg:disc}, ``Gauss-Hermite'' uses the Gauss-Hermite quadrature (with mean and standard deviation estimated by maximum likelihood), and ``NP-ME'' uses the maximum entropy method with the kernel density estimator. $T$ is the sample size in each simulation. $N$ is the number of nodes in the quadrature formula. $\gamma$ is the relative risk aversion. All results are based on 1,000 Monte Carlo replications.}
\end{table}

\begin{table}[!htb]
\centering
\caption{Relative mean absolute error of the optimal portfolio.}\label{t:MAE}
\begin{tabular}{ccccccccccc}
\toprule
\multicolumn{2}{c}{Method} & \multicolumn{3}{c}{NP-GQ} & \multicolumn{3}{c}{Gauss-Hermite} & \multicolumn{3}{c}{NP-ME}\\
$T$ & $N$ & $\gamma=2$ & 4 & 6 & $\gamma=2$ & 4 & 6 & $\gamma=2$ & 4 & 6 \\
\cmidrule(lr){1-2}
\cmidrule(lr){3-5}
\cmidrule(lr){6-8}
\cmidrule(lr){9-11}
\multirow{4}{*}{100} & 3 & 0.239&0.247&0.249&0.323&0.306&0.301&0.296&0.293&0.292\\
& 5 & 0.236&0.247&0.249&0.314&0.305&0.301&0.267&0.271&0.270\\
& 7 & 0.236&0.247&0.249&0.313&0.305&0.301&0.256&0.264&0.265\\
& 9 & 0.236&0.247&0.249&0.312&0.305&0.301&0.251&0.262&0.263\\
\cmidrule(lr){1-2}
\cmidrule(lr){3-5}
\cmidrule(lr){6-8}
\cmidrule(lr){9-11}
\multirow{4}{*}{1,000} & 3 & 0.068&0.072&0.073&0.125&0.100&0.095&0.112&0.098&0.095\\
& 5 &0.067&0.072&0.073&0.124&0.101&0.095&0.078&0.078&0.078\\
& 7 &0.067&0.072&0.073&0.124&0.101&0.095&0.074&0.076&0.076\\
& 9 &0.067&0.072&0.073&0.124&0.101&0.095&0.072&0.075&0.076\\
\cmidrule(lr){1-2}
\cmidrule(lr){3-5}
\cmidrule(lr){6-8}
\cmidrule(lr){9-11}
\multirow{4}{*}{10,000} & 3 & 0.021&0.023&0.023&0.098&0.056&0.045&0.084&0.056&0.048\\
& 5 &0.021&0.023&0.023&0.098&0.056&0.045&0.029&0.025&0.025\\
& 7 &0.021&0.023&0.023&0.098&0.056&0.045&0.024&0.024&0.024\\
& 9 &0.021&0.023&0.023&0.098&0.056&0.045&0.023&0.023&0.023\\
\bottomrule
\end{tabular}
\caption*{\footnotesize Note: the table reports the relative mean absolute error of the optimal portfolio defined by \eqref{eq:MAE}. See Table \ref{t:bias} for the definition of variables.}
\end{table}

\section{Gaussian quadrature}\label{sec:GQ}

In this appendix we prove some properties of the Gaussian quadrature. For notational simplicity let us omit $a,b$ (so $\int$ means $\int_a^b$) and assume that $\int w(x)x^n\diff x$ exists for all $n\ge 0$. For functions $f,g$, define the inner product $(f,g)$ by
\begin{equation}
(f,g)=\int_a^b w(x)f(x)g(x)\diff x.\label{eq:InnnerProduct}
\end{equation}
As usual, define the norm of $f$ by $\norm{f}=\sqrt{(f,f)}$. The first step is to construct orthogonal polynomials $\set{p_n(x)}_{n=0}^N$ corresponding to the inner product \eqref{eq:InnnerProduct}.

\begin{defn}[Orthogonal polynomial]
The polynomials $\set{p_n(x)}_{n=0}^N$ are called \emph{orthogonal} if
\begin{inparaenum}[(i)]
\item $\deg p_n=n$ and the leading coefficient of $p_n$ is 1, and
\item for all $m\neq n$, we have $(p_m,p_n)=0$.
\end{inparaenum}
\end{defn}
Some authors require that the polynomials are orthonormal, so $(p_n,p_n)=1$. In this paper we normalize the polynomials by requiring that the leading coefficient is 1, which is useful for computation. The following three-term recurrence relation (TTRR) shows the existence of orthogonal polynomials and provides an explicit algorithm for computing them.

\begin{prop}[Three-term recurrence relation, TTRR]\label{prop:TTRR}
Let $p_0(x)=1$, $p_1(x)=x-\frac{(xp_0,p_0)}{\norm{p_0}^2}$, and for $n\ge 1$ define
\begin{equation}
p_{n+1}(x)=\left(x-\frac{(xp_n,p_n)}{\norm{p_n}^2}\right)p_n(x)-\frac{\norm{p_n}^2}{\norm{p_{n-1}}^2}p_{n-1}(x).\label{eq:TTRR}
\end{equation}
Then $p_n(x)$ is the degree $n$ orthogonal polynomial.
\end{prop}

\begin{proof}
Let us show by induction on $n$ that
\begin{inparaenum}[(i)]
\item $p_n$ is an degree $n$ polynomial with leading coefficient 1, and
\item $(p_n,p_m)=0$ for all $m<n$.
\end{inparaenum}
The claim is trivial for $n=0$. For $n=1$, by construction $p_1$ is a degree 1 polynomial with leading coefficient 1, and since $p_0(x)=1$, we obtain
$$(p_1,p_0)=\left(\left(x-\frac{(xp_0,p_0)}{\norm{p_0}^2}\right)p_0,p_0\right)=(xp_0,p_0)-(xp_0,p_0)=0.$$

Suppose the claim holds up to $n$. Then for $n+1$, by \eqref{eq:TTRR} the leading coefficient of $p_{n+1}$ is the same as that of $xp_n$, which is 1. If $m=n$, then
\begin{align*}
(p_{n+1},p_n)&=\left(\left(x-\frac{(xp_n,p_n)}{\norm{p_n}^2}\right)p_n-\frac{\norm{p_n}^2}{\norm{p_{n-1}}^2}p_{n-1},p_n\right)\\
&=(xp_n,p_n)-(xp_n,p_n)-\frac{\norm{p_n}^2}{\norm{p_{n-1}}^2}(p_{n-1},p_n)=0.
\end{align*}
If $m=n-1$, then
\begin{align*}
(p_{n+1},p_{n-1})&=\left(\left(x-\frac{(xp_n,p_n)}{\norm{p_n}^2}\right)p_n-\frac{\norm{p_n}^2}{\norm{p_{n-1}}^2}p_{n-1},p_{n-1}\right)\\
&=(xp_n,p_{n-1})-\frac{(xp_n,p_n)}{\norm{p_n}^2}(p_n,p_{n-1})-\norm{p_n}^2\\
&=(p_n,xp_{n-1})-\norm{p_n}^2.
\end{align*}
Since the leading coefficients of $p_n,p_{n-1}$ are 1, we can write $xp_{n-1}(x)=p_n(x)+q(x)$, where $q(x)$ is a polynomial of degree at most $n-1$. Clearly $q$ can be expressed as a linear combination of $p_0,p_1,\dots,p_{n-1}$, so $(p_n,q)=0$. Therefore
$$(p_{n+1},p_{n-1})=(p_n,p_n+q)-\norm{p_n}^2=\norm{p_n}^2+(p_n,q)-\norm{p_n}^2=0.$$
Finally, if $m<n-1$, then
\begin{align*}
(p_{n+1},p_m)&=\left(\left(x-\frac{(xp_n,p_n)}{\norm{p_n}^2}\right)p_n-\frac{\norm{p_n}^2}{\norm{p_{n-1}}^2}p_{n-1},p_m\right)\\
&=(xp_n,p_m)-\frac{(xp_n,p_n)}{\norm{p_n}^2}(p_n,p_m)-\frac{\norm{p_n}^2}{\norm{p_{n-1}}^2}(p_{n-1},p_m)\\
&=(p_n,xp_m)=0
\end{align*}
because $xp_m$ is a polynomial of degree $1+m<n$.
\end{proof}

The following lemma shows that an degree $n$ orthogonal polynomial has exactly $n$ real roots (so they are all simple).

\begin{lem}\label{lem:simpleRoot}
$p_n(x)$ has exactly $n$ real roots on $(a,b)$.
\end{lem}

\begin{proof}
By the fundamental theorem of algebra, $p_n(x)$ has exactly $n$ roots in $\C$. Suppose on the contrary that $p_n(x)$ has less than $n$ real roots on $(a,b)$. Let $x_1,\dots,x_k$ ($k<n$) those roots at which $p_n(x)$ changes its sign. Let $q(x)=(x-x_1)\dotsb(x-x_k)$. Since $p_n(x)q(x)>0$ (or $<0$) almost everywhere on $(a,b)$, we have
$$(p_n,q)=\int w(x)p_n(x)q(x)\diff x\neq 0.$$
On the other hand, since $\deg q=k<n$, we have $(p_n,q)=0$, which is a contradiction.
\end{proof}

The following theorem shows that using the $N$ roots of the degree $N$ orthogonal polynomial $p_N(x)$ as quadrature nodes and choosing specific weights, we can integrate all polynomials of degree up to $2N-1$ exactly. Thus Gaussian quadrature always exists.

\begin{thm}[Gaussian quadrature]\label{thm:GQ}
Let $a<x_1<\dots<x_N<b$ be the $N$ roots of the degree $N$ orthogonal polynomial $p_N$ and define
$$w_n=\int w(x)L_n(x)\diff x$$
for $n=1,\dots,N$, where
$$L_n(x)=\prod_{m\neq n}\frac{x-x_m}{x_n-x_m}$$
is the degree $N-1$ polynomial that takes value 1 at $x_n$ and 0 at $x_m$ ($m\in \set{1,\dots,N}\backslash n$).
Then
\begin{equation}
\int w(x)p(x)\diff x=\sum_{n=1}^Nw_np(x_n)\label{eq:exactInt}
\end{equation}
for all polynomials $p(x)$ of degree up to $2N-1$.
\end{thm}

\begin{proof}
Since $\deg p\le 2N-1$ and $\deg p_N=N$, we can write
$$p(x)=p_N(x)q(x)+r(x),$$
where $\deg q,\deg r\le N-1$. Since $q$ can be expressed as a linear combination of orthogonal polynomials of degree up to $N-1$, we have $(p_N,q)=0$. Hence
$$\int w(x)p(x)\diff x=(p_N,q)+\int w(x)r(x)\diff x=\int w(x)r(x)\diff x.$$
On the other hand, since $\set{x_n}_{n=1}^N$ are roots of $p_N$, we have 
$$p(x_n)=p_N(x_n)q(x_n)+r(x_n)=r(x_n)$$
for all $n$, so in particular
$$\sum_{n=1}^Nw_np(x_n)=\sum_{n=1}^Nw_nr(x_n).$$
Therefore it suffices to show \eqref{eq:exactInt} for polynomials $r$ of degree up to $N-1$. Let us show that
$$r(x)=\sum_{n=1}^Nr(x_n)L_n(x)$$
identically. To see this, let $\tilde{r}$ be the right-hand side. Since $L_n(x_m)=\delta_{mn}$ (Kronecker's delta), we have
$$\tilde{r}(x_m)=\sum_{n=1}^Nr(x_n)L_n(x_m)=\sum_{n=1}^N\delta_{mn}r(x_n)=r(x_m),$$
so $r$ and $\tilde{r}$ agree on $N$ distinct points $\set{x_n}_{n=1}^N$. Since each $L_n(x)$ is a degree $N-1$ polynomial, we have $\deg \tilde{r}\le N-1$. Therefore it must be $r=\tilde{r}$.

Since $r$ can be represented as a linear combination of $L_n$'s, it suffices to show \eqref{eq:exactInt} for all $L_n$'s. But since by definition
$$\int w(x)L_n(x)\diff x=w_n=\sum_{m=1}^Nw_m\delta_{mn}=\sum_{m=1}^Nw_mL_n(x_m),$$
the claim is true.
\end{proof}

In practice, how can we compute the nodes $\set{x_n}_{n=1}^N$ and weights $\set{w_n}_{n=1}^N$ of the $N$-point Gaussian quadrature? The solution is given by the following Golub-Welsch algorithm.

\begin{thm}[\citealp{GolubWelsch1969}]\label{thm:GolubWelsch}
For each $n\ge 1$, define $\alpha_n,\beta_n$ by $$\alpha_n=\frac{(xp_{n-1},p_{n-1})}{\norm{p_{n-1}}^2}, \quad \beta_n=\frac{\norm{p_n}}{\norm{p_{n-1}}}>0.$$
Define the $N\times N$ symmetric tridiagonal matrix $T_N$ as in \eqref{eq:TN}. Then the Gaussian quadrature nodes $\set{x_n}_{n=1}^N$ are eigenvalues of $T_N$. Letting $v_n=(v_{n1},\dots,v_{nn})'$ be an eigenvector of $T_N$ corresponding to eigenvalue $x_n$, then the weights $\set{w_n}_{n=1}^N$ are given by
\begin{equation}
w_n=\frac{v_{n1}^2}{\norm{v_n}^2}\int w(x)\diff x>0.\label{eq:GaussWeight}
\end{equation}
\end{thm}

\begin{proof}
By \eqref{eq:TTRR} and the definition of $\alpha_n,\beta_n$, for all $n\ge 0$ we have
$$p_{n+1}(x)=(x-\alpha_{n+1})p_n(x)-\beta_n^2 p_{n-1}(x).$$
Note that this is true for $n=0$ by defining $p_{-1}(x)=0$ and $\beta_0=0$. For each $n$, let $p_n^*(x)=p_n(x)/\norm{p_n}$ be the normalized orthogonal polynomial. Then the above equation becomes
$$\norm{p_{n+1}}p_{n+1}^*(x)=\norm{p_n}(x-\alpha_{n+1})p_n^*(x)-\norm{p_{n-1}}\beta_n^2p_{n-1}^*(x).$$
Dividing both sides by $\norm{p_n}>0$, using the definition of $\beta_n,\beta_{n+1}$, and rearranging terms, we obtain
$$\beta_np_{n-1}^*(x)+\alpha_{n+1}p_n^*(x)+\beta_{n+1}p_{n+1}^*(x)=xp_n^*(x).$$
In particular, setting $x=x_k$ (where $x_k$ is a root of $p_N$), we obtain
$$\beta_np_{n-1}^*(x_k)+\alpha_{n+1}p_n^*(x_k)+\beta_{n+1}p_{n+1}^*(x_k)=x_kp_n^*(x_k).$$
for all $n$ and $k=1,\dots,N$. Since $\beta_0=0$ by definition and $p_N^*(x_k)=0$ (since $x_k$ is a root of $p_N$ and hence $p_N^*=p_N/\norm{p_N}$), letting $P(x)=(p_0^*(x),\dots,p_{N-1}^*(x))'$ and collecting the above equation into a vector, we obtain
$$T_NP(x_k)=x_kP(x_k)$$
for $k=1,\dots,N$. Define the $N\times N$ matrix $P$ by $P=(P(x_1),\dots,P(x_N))$. Then $T_NP=\diag(x_1,\dots,x_N)P$, so $x_1,\dots,x_N$ are eigenvalues of $T_N$ provided that $P$ is invertible. Now since $\set{p_n^*}_{n=0}^{N-1}$ are normalized and Gaussian quadrature integrates all polynomials of degree up to $2N-1$ exactly, we have
$$\delta_{mn}=(p_m^*,p_n^*)=\int w(x)p_m^*(x)p_n^*(x)\diff x=\sum_{k=1}^Nw_kp_m^*(x_k)p_n^*(x_k)$$
for $m,n\le N-1$. Letting $W=\diag(w_1,\dots,w_N)$, this equation becomes $PWP'=I$. Therefore $P,W$ are invertible and $x_1,\dots,x_N$ are eigenvalues of $T_N$. Solving for $W$ and taking the inverse, we obtain
$$W^{-1}=P'P\iff \frac{1}{w_n}=\sum_{k=0}^{N-1}p_k^*(x_n)^2>0$$
for all $n$. To show \eqref{eq:GaussWeight}, let $v_n$ be an eigenvector of $T_N$ corresponding to eigenvalue $x_n$. Then $v_n=cP(x_n)$ for some constant $c\neq 0$. Taking the norm, we obtain
$$\norm{v_n}^2=c^2\norm{P(x_n)}^2=c^2\sum_{k=0}^{N-1}p_k^*(x_n)^2=\frac{c^2}{w_n}\iff w_n=\frac{c^2}{\norm{v_n}^2}.$$
Comparing the first element of $v_n=cP(x_n)$, noting that $p_0(x)=1$ and hence $p_0^*=p_0/\norm{p_0}=1/\norm{p_0}$, we obtain
$$c^2=v_{n1}^2\norm{p_0}^2=v_{n1}^2\int w(x)p_0(x)^2\diff x=v_{n1}^2\int w(x)\diff x,$$
which implies \eqref{eq:GaussWeight}.
\end{proof}

\end{document}